\documentclass[final,12pt]{article}
\usepackage[utf8]{inputenc}
\usepackage[T1]{fontenc}

\usepackage{ifthen}
\usepackage{hyperref}

\ifdefined \bDoNotIncludePackages \else
  \newcommand{\bDoNotIncludePackages}{0}
\fi
\ifdefined \bSkipDocumentSetting \else
  \newcommand{\bSkipDocumentSetting}{0}
\fi
\ifdefined \bDoNotDefineTheorems \else
  \newcommand{\bDoNotDefineTheorems}{0}
\fi

\ifthenelse{\equal{\bDoNotIncludePackages}{1}}{}
{
    \usepackage{amsmath,amsthm,amssymb}
    \usepackage{cleveref} 

    \usepackage[affil-it]{authblk}
}

\ifthenelse{\equal{\bSkipDocumentSetting}{1}}{}{

\addtolength{\voffset}{-1cm} 
\addtolength{\hoffset}{-1.5cm} 
\setlength{\textheight}{22cm} \setlength{\textwidth}{16cm}
}

\def\Z{\mathbb Z}
\def\R{\mathbb R}
\def\Q{\mathbb Q}
\def\N{\mathbb N}
\def\A{\mathcal A}

\def\S{\mathcal S}

\def\tt{\mathbf t}
\def\aa{\mathbf a}

\def\GG{G}
\def\SL{ {\rm SL} }

\ifthenelse{\equal{\bDoNotDefineTheorems}{1}}{}
{
\newtheorem{thm}{Theorem}
\newtheorem{theorem}[thm]{Theorem}

\newtheorem{corollary}[thm]{Corollary}

\newtheorem{lemma}[thm]{Lemma}

\newtheorem{proposition}[thm]{Proposition}

\theoremstyle{definition}

\newtheorem{definition}[thm]{Definition}
\newtheorem{example}[thm]{Example}

\crefname{thm}{theorem}{theorems}
\crefname{thrm}{theorem}{theorems}
\crefname{coro}{corollary}{corollaries}
\crefname{example}{example}{examples}
\crefname{lem}{lemma}{lemmas}
\crefname{lmm}{lemma}{lemmas}
\crefname{claim}{claim}{claims}
\crefname{obs}{observation}{observations}
\crefname{proposition}{proposition}{propositions}
\crefname{prop}{proposition}{propositions}
\crefname{defi}{definition}{definitions}

\crefname{theorem}{theorem}{theorems}
\crefname{corollary}{corollary}{corollaries}
\crefname{example}{example}{examples}
\crefname{lemma}{lemma}{lemmas}
\crefname{proposition}{proposition}{propositions}
\crefname{definition}{definition}{definitions}

\theoremstyle{remark}
\newtheorem{remark}[thm]{Remark}

\crefname{example}{example}{examples}
}

\begin{document}

\title{Markov constant and quantum instabilities}

\author{Edita Pelantová}
\affil{Faculty of Nuclear Sciences and Physical Engineering\\ Czech Technical University in Prague\\ Prague, Czech Republic}

\author{Štěpán Starosta%
  \thanks{Electronic address: \texttt{stepan.starosta@fit.cvut.cz}; Corresponding author}}
\affil{Faculty of Information Technology\\ Czech Technical University in Prague\\ Prague, Czech Republic}

\author{Miloslav Znojil}
\affil{Nuclear Physics Institute ASCR\\ Řež, Czech Republic}

\maketitle

\begin{abstract}
For a qualitative analysis of spectra of certain two-dimensional rectangular-well quantum systems several rigorous methods of number
theory are shown productive and useful.
These methods (and, in particular, a generalization of the concept of Markov constant known in Diophantine approximation theory) are shown to provide a new mathematical insight in the phenomenologically relevant occurrence of anomalies in the spectra.
Our results may inspire methodical innovations ranging from the description of the stability properties of metamaterials and of certain hiddenly unitary quantum evolution models up to the clarification of the mechanisms of occurrence of ghosts in quantum cosmology.
\end{abstract}


\section{Introduction} \label{zacatek}

The main {\em mathematical} inspiration of our present physics-oriented paper may be traced back to the theory of Diophantine approximations in which an important role is played by certain sets of real numbers possessing an accumulation point called Markov constant \cite{Burger-jungle}.
The related ideas and techniques (to be shortly outlined below) are transferred to an entirely different context.
Briefly, we show that and how some of the results of number theory may appear applicable in an analysis of realistic quantum dynamics.

The sources of our {\em phenomenological} motivation are more diverse.
Among them, a distinct place is taken by the problems of quantum stability which are older than the quantum theory itself.
Their profound importance already became clear in the context of the Niels Bohr's model of atom \cite{Bohr}.
In this light one of the main achievements of the early quantum theory may be seen precisely in the explanation of the well verified experimental observation that many quantum systems (like hydrogen atom, etc.) are safely stable.

During the subsequent developments of the quantum theory, the rigorous mathematical foundation of the concept of quantum stability found its safe ground in the spectral theory of self-adjoint operators in Hilbert space \cite{Simon}.
Although it may sound like a paradox, a similar interpretation of the {\em loss} of quantum stability is much less developed at present. This does not imply that the systematic study of instabilities would be less important.
The opposite is true because the majority of existing quantum systems ranging from elementary particles to atomic nuclei and molecules are unstable.

In this direction of study one could only feel discouraged by the fact that the existing theoretical descriptions of quantum instabilities require complicated mathematics, be it in quantum field theory, in statistical quantum physics or, last but not least, in the representations of quantum models using non-selfadjoint operators \cite{Babook}.
For this reason we believe that our present approach combining a sufficiently rigorous level of mathematics with a not too complicated exemplification of quantum systems might offer a fresh and innovative perspective to quantum physics and, in
particular, to some of its stability and instability aspects.

It is certainly encouraging for us to notice that a combination of Diophantine analysis with phenomenological physics already appeared relevant in the context of study of certain stable quantum systems controlled by point interactions and living on rectangular lattices \cite{mrizka} or on hexagonal lattices \cite{hexe} where, typically, the band spectra may depend on certain number-theoretical characteristics of the system.
In what follows, we intend to turn our attention from complicated quantum graphs to a maximally elementary and exactly solvable model in which the hyperbolic partial differential equation
 \begin{equation}
\Box f(x,y) = \lambda f(x,y)\,, \ \ \ \
 \Box= \frac{\partial^2}{\partial x^2} - \frac{\partial^2}{\partial y^2}
 \,, \ \ \ \ f|_{\partial R} = 0
 \label{presm}
  \end{equation}
is studied and in which the instability is immanently present, in a way to be discussed below, via the unboundedness of the spectrum from below.

In our model the eigenfunctions are required to satisfy the most common Dirichlet boundary conditions, i.e., they are expected to vanish along the boundary of the two-dimensional rectangle
\begin{equation}
R = \{ (x,y) \colon 0 \leq x \leq a, 0 \leq y \leq b \}.
\label{represm}
 \end{equation}
In sections \ref{spepro}--\ref{posledni-matematicka} we describe and prove rigorous results of analysis of such a model.
After a systematic presentation of these mathematical observations we return, in sections \ref{konec} and \ref{ukonec}, to the problem of their various potential connections with physics.
We also list there a few not entirely artificial samples of placing the Klein-Gordon-resembling Eq.~(\ref{presm}) into a broader phenomenological context.

\section{Spectral problem} \label{spepro}

\subsection{Separation of variables}

Our present analysis is fully concentrated upon the properties of spectra of hyperbolic partial differential operators $\Box$ of Eqs.~(\ref{presm}) + (\ref{represm}) which act upon twice differentiable functions $ f(x,y)  $ of two real variables.
Setting $f(x,y) = g(x) h(y)$ we find that the eigenvalue problem is easily solvable by separation of variables, i.e., that there exist
constants $C$ and $D$ such that
$$
\frac{\Box f(x,y)}{f(x,y)} = \frac{\frac{\partial^2 g(x)}{\partial
x^2}}{g(x)}
 - \frac{\frac{\partial^2 h(y)}{\partial y^2}}{h(y)} =  C - D = \lambda.
$$
The solution of the corresponding ordinary differential equation for unknown $g(x)$ (and, {\it mutatis mutandis}, for $h(y)$) yields
$$
g(x) = \alpha \sin (\sqrt{-C} x) + \beta \cos (\sqrt{-C} a)
$$
for $C< 0$,
$$
g(x) = \alpha x + \beta
$$
for $C = 0$, and
$$
g(x) = \alpha e^{-\sqrt{C} x} + \beta e^{\sqrt{C}x}
$$
for $C > 0$.
Under our Dirichlet boundary conditions, a nonzero solution is obtained only for $C < 0$. We obtain
$$
a\sqrt{-C} = m \pi
$$
for $m \in \Z$.
Analogously, we obtain
$$
b\sqrt{-D} = k \pi
$$
for $k \in \Z$.
Since $\lambda = C - D$, we have, finally,
 $$
 \lambda_{k,m} = \frac{k^2\pi^2}{a^2} - \frac{m^2\pi^2}{b^2}
 = \frac{\pi^2m^2}{a^2}\left( \frac{k^2}{m^2} -
 \frac{a^2}{b^2}\right) = \frac{\pi^2m^2}{a^2}\left( \frac{k}{m}
 - \frac{a}{b}\right)\left( \frac{k}{m} + \frac{a}{b}\right)
 $$
for all $k,m \in \Z$.
Thus, the spectrum equals the closure of the set of all $\lambda_{k,m}$:
\[\sigma(\Box)=\overline{ \{ \lambda_{k,m} \colon k,m \in \Z
\}}.\]

\subsection{The number theory approach}

Up to a multiplicative factor, the singular part of the spectrum  $\sigma(\Box)$ coincides with the set
 $$
 \S(\alpha) = \text{set\ of\ all\ accumulation\ points\ of\ }
 \left \{ m^2 \left( \frac{k}{m} - \alpha \right)
 \colon k,m \in \Z  \right \}
 $$
where the ratio $\alpha = {a}/{b}$ is a dynamical parameter of the model.
The structure of such sets is well understood in the theory of Diophantine approximations.
In particular, the smallest accumulation point of the displayed set  -- the so-called  Markov constant of $\alpha$ --  is in the centre of interest of many mathematicians.

This observation is in fact a methodical starting point of our present paper.
In essence, our analysis of stability/instability issues are mainly inspired by the results of the existing number-theory literature on Markov constant.

\section{Simple properties of $\S(\alpha)$}

Assume $\alpha \in \R$.
As the set $\Z^2$ is countable, the set $\left \{ m^2 \left( \frac{k}{m} - \alpha \right) \colon k,m \in \Z  \right \}$ can be viewed as the range of a real sequence.
Let us rephrase the definition of $\S(\alpha)$: a number $x$ belongs to  $ \S(\alpha)$ if there exist strictly
monotone sequences of integers $(k_n)$  and $(m_n)$ such that $x=\lim\limits_{n\to \infty} m_n^2\Bigl( \frac{k_n}{m_n} - \alpha\Bigr)$.

We list several simple properties of $\S(\alpha)$.

\begin{enumerate}

\item Since the set of accumulation points of any real sequence is closed, the set $\S(\alpha)$ is a topologically closed subset of $\R$.

\item $\S(\alpha)$ is closed under multiplication by $z^2$ for each $z \in \Z$.

\begin{proof}
If $x \in \S(\alpha)$, i.e., $m_n^2\Bigl( \frac{k_n}{m_n} - \alpha\Bigr) \to x$, then $(m_nz)^2\Bigl( \frac{k_nz}{m_nz} - \alpha\Bigr) \to xz^2$, thus $xz^2 \in \S(\alpha)$.
\end{proof}

\item  If $\alpha \in \Q$,  then  $\S(\alpha)$ is empty.

\begin{proof}
If $\alpha =\frac{r}{s}$ with $r,s \in \Z$, then  $m^2\Bigl(\frac{k}{m} - \frac{r}{s}\Bigr) =  \frac{t}{s}$ for some $t\in \Z$.
It means that $\left\{ m^2 \left( \frac{k}{m} - \alpha \right) \colon k,m \in \Z  \right \}$  is a subset of the  discrete set $\frac{1}{s} \Z$.
\end{proof}

\item If $ \alpha \notin \Q$,  then  $\S(\alpha)$ has at least one element in the interval $[-1 ,1]$.

\begin{proof}
According to Dirichlet's theorem, there exist infinitely many rational numbers $\frac{k}{m}$ such that $\bigl|\frac{k}{m} - \alpha  \bigr|< \frac{1}{m^2}$.
\end{proof}
\end{enumerate}

In order to present another remarkable property of $\S(\alpha)$ we exploit simple rational transformations connected with
$$
\GG= \left \{  g \in \Z^{2\times 2}  \colon \det (g) \neq 0 \right \}\quad \text{ and } \quad \SL_2( \Z )= \left \{  g  \in \GG  \colon \det (g) =1 \right \}.
$$
Note that  $\GG$ is a monoid, whereas  $\SL_2( \Z )$ is a group.
We define the action of $g = \begin{pmatrix} c & d \\ e & f \end{pmatrix} \in \GG $  on  the set $\R$ by  $\alpha \mapsto g\alpha = \frac{c\alpha+d}{e\alpha + f}$.

\begin{proposition}\label{mobius}
Let $\alpha \in \R$ and $g \in \GG$.
We have
$$
\det (g) \S(\alpha) \subset    \S(g\alpha).
$$
In particular, $\S(g\alpha) = \S(\alpha)$ if $g \in  \SL_2( \Z )$.
\end{proposition}

\begin{proof}
Let $g = \begin{pmatrix} c & d \\ e & f \end{pmatrix} \in \GG$.
Let $x \in \S(\alpha)$ and let $(k_n)$ and $(m_n)$ be sequences such that $m_n^2 \left( \frac{k_n}{m_n} - \alpha \right) \to x$.
We set
$$
k'_n = ck_n + dm_n \quad \text{ and } \quad m'_n = ek_n + fm_n.
$$

We obtain
\begin{align*}
\left(\frac{k'_n}{m'_n}- g \alpha   \right) {m'_n}^2 & =  \left( \frac{ck_n + dm_n}{ek_n + fm_n}- \frac{c\alpha+d}{e\alpha + f}  \right) (ek_n + fm_n)^2 \\
& = \frac{k_n (cf - de)-\alpha m_n (cf - de)}{(e\alpha + f)(ek_n + fm_n)} (ek_n + fm_n)^2 \\
& = \frac{(k_n -\alpha m_n  )(cf - de) }{e\alpha + f} (ek_n + fm_n) \\
& = \det (g)\,  m_n^2 \Bigl(\frac{k_n}{m_n}-\alpha\Bigr) \, \frac { e\frac{k_n}{m_n} + f}{e\alpha + f} \\
& \underset{n \to + \infty}{\longrightarrow} \det (g) \,x
\end{align*}
as $\frac{k_n}{m_n} \to \alpha$.
It means that $\det(g)\,x$ belongs to $\S(g\alpha)$.\\

If  $\det g =1$, then $\S(\alpha) \subset \S(g\alpha)$ and  $g^{-1} \in  \SL_2( \Z )$ as well.
Therefore,  $\S(g\alpha) \subset \S(g^{-1} g\alpha) =\S(\alpha)$, too.
\end{proof}

In the sequel,  $\lfloor x\rfloor$ stands for  the integer part of $x$, i.e., the  largest integer $n$ such that $n\leq x$.
Since  $ \alpha- \lfloor\alpha\rfloor = g \alpha$, where $g = \begin{pmatrix} 1 & \lfloor\alpha\rfloor  \\ 0 & 1 \end{pmatrix} \in \SL_2( \Z )$, the previous proposition immediately implies the following corollary.
\begin{corollary}
For any  $\alpha\in \R$ we have   $\S(\alpha) = \S(\alpha- \lfloor\alpha\rfloor)$.
\end{corollary}
Note that   $g\alpha =\alpha$ for any   $g = \begin{pmatrix} z& 0 \\ 0 & z\end{pmatrix}$ with $z\in \Z$.
\Cref{mobius} implies $z^2\S(\alpha) \subset \S(\alpha)$, as already observed.

\section{Continued fractions and convergents} \label{ctyri}
The theory of continued fractions plays a crucial role in Diophantine approximation, i.e., in  approximation of an irrational number by a rational number.
The definition of $\S(\alpha)$ indicates that the quality of an approximation of $\alpha$ by fractions $\frac{k}{m}$ governs the behaviour of  $ \S(\alpha)$.
The continued fraction of an irrational  number $x$ is a coding of the orbit of $x$ under a transformation $T$ defined by
$$T: \mathbb{R}\setminus \Q  \ \to \  (1, +\infty)\setminus \Q\qquad \text{ and } \qquad  T(x) = \frac{1}{x-\lfloor x\rfloor }.$$

\begin{definition}
Let $x \in \mathbb{R}\setminus \Q $.
The \textit{continued fraction} of $x$ is the infinite sequence of integers $[a_0, a_1, a_2, a_3, \ldots]$   where
$$a_i=  \lfloor T^i(x)\rfloor \quad \text{ for all } \  i=0,1,2,3, \ldots\,.$$
\end{definition}
 Clearly, for all $i\geq 1$ the coefficient $a_i$ is a positive integer.
 Only the coefficient $a_0$ takes values in the whole range of integers.

 If $\alpha$ is an irrational number, then $T(\alpha) = g\alpha$ with    $g = \begin{pmatrix} 0 & 1 \\ 1 & -\lfloor\alpha\rfloor  \end{pmatrix} \in \GG$.
 As $\det(g) = -1$, Proposition \ref{mobius} implies $  \S(T(\alpha)) = - \S(\alpha).$
A number $x$ is usually identified with its continued fraction and we also write $x=[a_0, a_1, a_2, a_3, \ldots]$.
Using this convention, the previous fact can be generalized as follows.

\begin{corollary}\label{stejne} Let $[a_0, a_1, a_2, a_3, \ldots]$ be a continued fraction.
We have
$$
\S([a_{n+k}, a_{n+1+k}, a_{n+2+k}, \ldots]) = (-1)^k \S([a_n, a_{n+1}, a_{n+2},  \ldots])\, \text{ for any  \  }
k,n \in \N.
$$
\end{corollary}

The knowledge of the continued fraction of $x$ allows us to find the best rational approximations, in a certain sense, of the number $x$.
To describe these approximations, we use the following notation:
$[a_0, a_1, a_2, a_3, \ldots, a_n]$, where $a_0 \in \Z$ and $a_1, \ldots, a_n \in \N \setminus\{0\}$, denotes the fraction
$$
a_0 + \cfrac{1}{a_1 + \cfrac{1}{a_2 + \cfrac{1}{\ddots  + \cfrac{1}{a_n}}}} \ .
$$
The number $a_i$ is said to be the $i$-th \textit{partial quotient} of $x$.

\begin{definition} Let $x$ be an irrational number, $[a_0, a_1, a_2, a_3, \ldots]$ its continued fraction and let  $N \in \N$.  Let $p_N \in \Z$  and $q_N\in \N \setminus\{0\}$ denote coprime numbers such that
$\frac{p_N}{q_N} = [a_0, a_1, a_2, a_3, \ldots, a_N]$.
The fraction $\frac{p_N}{q_N}$ is called the \textit{$N^{th}$-convergent} of $x$.
 \end{definition}

We list the relevant properties of convergents of an irrational number $\alpha$.
They can be found in any textbook of number theory, for example \cite{Burger-jungle}.
\begin{enumerate}

\item  We have $p_0 =a_0,  q_0=1,  p_1 = a_0a_1+1$ and $q_1=a_1$.
For any $N \in \N$, we have
\begin{equation} \label{recurence}
p_{N+1} = a_{N+1}p_{N}+ p_{N-1} \quad \text{and} \quad q_{N+1} = a_{N+1}q_{N}+ q_{N-1}.
\end{equation}

\item  For $N\in\N$, set $\alpha_{N+1} = [a_{N+1},a_{N+2},a_{N+3},\ldots ]$.
We have \begin{equation}\label{distance}\frac{p_N}{q_N}-\alpha   = \frac{(-1)^{N+1}}{q_{N}(\alpha_{N+1}q_N +q_{N-1})}\,, \quad \text {and in particular, } \ \ \left| \frac{p_N}{q_N}-\alpha \right| < \frac{1}{q_N^2}. \end{equation}

\item For $N\in\N$ and $a\in \Z$ satisfying $1\leq a \leq a_{N+1}-1$ we have
\begin{equation}\label{secondconvergents}\frac{ap_N+p_{N-1}}{aq_N+q_{N-1}}  -\alpha = \frac{(-1)^{N+1}(\alpha_{N+1}-a)}{(aq_{N}+q_{N-1})(\alpha_{N+1}q_N +q_{N-1})}.
\end{equation}
These rational approximations are known as {\it secondary convergents} of $\alpha$.
\end{enumerate}

\begin{corollary}
Let $\alpha$ be  an irrational number and $I$ be an interval.
There exists $\beta \in I$ such that $ \S(\alpha) = \S(\beta)$.
\end{corollary}
\begin{proof}
Without loss of generality, let $I$ be an open interval  and $\gamma \in I$ be an irrational number.
Let  $[a_0,a_1, a_2, \ldots ]$ and $[c_0,c_1, c_2, \ldots ]$ be continued fractions of $\alpha$ and $\gamma$, respectively.
Find $\varepsilon >0$ such that  $(\gamma -2\varepsilon, \gamma + 2\varepsilon) \subset I$.
In virtue of \eqref{distance} one can find  an integer $N$ such that the $N^{th}$-convergent $\frac{p'_{N}}{q'_{N}} $  of $\gamma$ satisfies $ \left| \gamma - \frac{p'_{N}}{q'_{N}}\right| <  \frac{1}{(q'_{N})^2}<\varepsilon$.
Define $$\beta=[c_0,c_1, \ldots, c_{N}, a_{N+1},a_{N+2}, \ldots ].$$
As the $N^{th}$-convergents  of  $\beta$ and $\gamma$ coincide and due to \eqref{distance}, we have \ $|\gamma - \beta| <  \frac{2}{(q'_{N})^2}<2\varepsilon$ and thus $\beta\in I$.
\Cref{stejne} implies that  the sets $\S(\alpha)$ and  $\S(\beta)$ coincide as well.
\end{proof}

\begin{theorem} \label{th:12-12}
Let $\alpha$ be an irrational number and $(\frac{p_N}{q_N})_{  N \in \N}$ be the sequence of its convergents.
If $x$ belongs to  $\S(\alpha)\cap (-\tfrac12, \tfrac12)$,  then $x$ is an accumulation point of the sequence
\begin{equation}\label{tojeono}
\left ( q_N^2 \left( \frac{p_N}{q_N} - \alpha \right)\right)_{  N \in \N}.
\end{equation}
\end{theorem}
\begin{proof} The theorem is a direct consequence of Legendre's theorem (see for instance \cite{Burger-jungle}, Theorem 5.12): Let $\alpha$ be an irrational number and $\frac{p}{q} \in \Q$.
If
$
\left | \frac{p}{q} - \alpha \right | < \frac{1}{2q^2},
$
then $\frac{p}{q}$ is a convergent of $\alpha$.
\end{proof}

Therefore, we start to investigate the accumulation points of the sequence \eqref{tojeono}.

\begin{lemma}
\label{odhad}
Let  $\alpha$ be an irrational number and $(\frac{p_N}{q_N})_{  N \in \N}$ be the sequence of its convergents.
For any $N\in \N$ we have
$$
q_N^2 \left( \frac{p_N}{q_N} - \alpha \right) = (-1)^{N+1} \Bigl( [a_{N+1}, a_{N+2},\ldots] + [0,a_N,a_{N-1},\ldots,a_1]  \Bigr)^{-1}.
$$
In particular,  for any $N\in \N$
$$\frac{1}{2+ a_{N+1}} < \left| q_N^2 \left( \frac{p_N}{q_N} - \alpha \right) \right| <  \frac{1}{a_{N+1}}.$$
\end{lemma}

\begin{proof}
Using \eqref{distance} we obtain
 $$ q_N^2 \left( \frac{p_N}{q_N} - \alpha \right)= \frac{(-1)^{N+1}}{\alpha_{N+1}+\frac{q_{N-1}}{q_{N}}}.$$
 By definition $\alpha_{N+1} = [a_{N+1}, a_{N+2},\ldots]$.
 It remains to show that $\frac{q_{N-1}}{q_{N}} = [0,a_{N}, a_{N-1}\ldots, a_1]$.
 We exploit the recurrent relation \eqref{recurence} for $(q_N)$.
 We proceed by induction:

If $N=1$, then $q_0 = 1$ and $q_1=a_1$.
Clearly $ \frac{q_{0}}{q_{1}} = \tfrac{1}{a_1} = [0,a_1]$.

  If $N>1$, then
  \begin{equation}\label{prevod}\frac{q_{N-1}}{q_{N}} = \frac{q_{N-1}}{a_{N}q_{N-1}+q_{N-2}} = \frac{1}{a_N + \frac{q_{N-2}}{q_{N-1}}}.
  \end{equation}
The number $\beta \in (0,1)$ has its continued fraction  in the form $[0,b_1,b_2,\ldots]$.
If  $1\leq B\in  \Z$,  then   the algorithm for construction of continued fraction assigns to the number $\frac{1}{B+\beta}$ the continued fraction $[0,B,b_1,b_2, \ldots]$.
We apply this rule and the induction assumption to \eqref{prevod} with $B=a_N$ and $\beta =  \frac{q_{N-2}}{q_{N-1}} = [0,a_{N-1},a_{N-2},\ldots,a_1]$.
\end{proof}

\subsection{Spectra of  quadratic numbers}

A famous theorem of Lagrange says that an irrational  number $\alpha$ is a root of the quadratic polynomial $Ax^2+Bx+C$  with integer coefficients $A,B,C$  if and only if the continued fraction of $\alpha$ is eventually periodic, i.e.,  $\alpha = [a_0,a_1, \ldots, a_s,  (a_{s+1}, \ldots, a_{s+\ell})^\omega]$, where  $v^\omega$ denotes the infinite string formed by the repetition of the finite string $v$.

\begin{theorem} \label{th:quadratic}
Let $\alpha$ be a quadratic number and $\bigl(\frac{p_N}{q_N}\bigr)$ be the sequence if its convergents.
Let $\ell$ be the smallest period of the repeating part of the continued fraction of $\alpha$.
The sequence $\left ( q_N^2 \bigl( \frac{p_N}{q_N} - \alpha \bigr)\right)_{  N \in \N}$ has at most
\begin{itemize}
\item $\ell$ accumulation points if $\ell$ is even;
\item $2\ell$ accumulation points  if $\ell$ is odd.
\end{itemize}
Moreover, at least one of the accumulation points belongs to the interval $\left (-\tfrac12, \tfrac12 \right )$.
\end{theorem}

\begin{proof}
According to \Cref{stejne} we can assume that the continued fraction of $\alpha$ is purely periodic, i.e., $\alpha = [(a_0,a_1, \ldots,a_{\ell -1})^\omega]$ for some $\ell > 0$, and that the first digit satisfies $a_0 = \max\{a_0,a_1, \ldots, a_{\ell -1}\}$.
Let $D$ denote the set of the accumulation points of $\left ( q_N^2 \bigl( \frac{p_N}{q_N} - \alpha \bigr)\right)_{  N \in \N}$.

Suppose $\ell$ is even.
By \Cref{odhad} and since $\alpha = [(a_0,a_1, \ldots,a_{\ell -1})^\omega]$, it follows that all the elements of $D$ are the limit-points of the sequences $(c^{(j)}_{k})_{k \in \N}$ where
$$
c^{(j)}_{k} = (-1)^{j+k\ell-1} \left ( [a_{j+k\ell}, a_{j+k\ell+1},\ldots] + [0,a_{j+k\ell-1},\ldots,a_1]  \right)^{-1}
$$
for each $j$ with $0 \leq j < \ell$.
As $\ell$ is even, the term $(-1)^{j+k\ell-1}$ equals $(-1)^{j-1}$ and a limit exists.
Thus, $\# D \leq \ell$.

If $\ell$ is odd, we define the number $c^{(j)}_{k}$ for $0 \leq j < 2\ell$ in the same way and the elements of $D$ are exactly the limit-points of the sequences $(c^{(j)}_{2k})_{k \in \N}$.
The term $(-1)^{j+2k\ell-1}$ in the expression of $c^{(j)}_{2k}$ equals again $(-1)^{j-1}$ and a limit exists for all $j$.
Thus, $\# D \leq 2\ell$.

If $a_0 = 1$, then $\alpha = [1^\omega]$, i.e., it is the golden ratio.
We have
$$
\lim_{k \to +\infty} c^{(0)}_{2k} = - \left( [1^\omega] + [0,1^\omega] \right )^{-1} = - \left( \frac{1 + \sqrt 5}{2} + \frac{2}{1 + \sqrt 5} \right )^{-1} = -\frac{1}{\sqrt 5} \geq - \frac{1}{2}.
$$
Thus, in this case, $D \cap \left(-\frac12,\frac12 \right)$ is not empty.

If $a_0 = 2$, then
$$
\left | c^{(0)}_{k} \right | = \Big | [2, a_{k\ell+2},\ldots] + [0,a_{k\ell-1},\ldots,a_1]  \Big | ^{-1} =  \Big | 2 + [0, a_{k\ell+2},\ldots] + [0,a_{k\ell-1},\ldots,a_1]  \Big | ^{-1} \leq \frac12.
$$
It implies that $D \cap \left(-\frac12,\frac12 \right)$ is not empty.
\end{proof}

We add some remarks on the last theorem.
The following observation follows from the last proof: if $\ell$ is odd, then $D$ is symmetric around $0$.

Let $\eta^{(j)} = [( a_{j}, \ldots, a_{j + \ell - 1})^\omega]$.
The number $\eta^{(j)}$ is a reduced quadratic surd and its conjugate $\tilde{\eta}^{(j)}$ satisfies
$$
-\frac{1}{\tilde{\eta}^{(j)}} = [( a_{j + \ell - 1}, \ldots, a_{j})^\omega].
$$
Therefore,
\begin{equation} \label{eq:formula_quadratic}
\lim_{k \to +\infty} c^{(j)}_{2k} = \frac{(-1)^{j-1}}{ \eta^{(j)} - \tilde{\eta}^{(j)}}.
\end{equation}

As follows from the last proof, the bound of \Cref{th:quadratic} is tight.
On the other hand, there exist quadratic numbers such that the bound is not attained.
It suffices to set $\alpha = [(1,2,1,1)^\omega]$.
We have
\begin{align*}
[(1,2,1,1)^\omega] &= \frac25 \sqrt 6 + \frac25 \quad \text{ and } \\
[(1,1,1,2)^\omega] &= \frac25 \sqrt 6 + \frac35 .
\end{align*}
Using \eqref{eq:formula_quadratic} we obtain $\# D < \ell = 4$.
In fact,  $D = \left \{ - \frac{5}{4 \sqrt 6}, \frac{1}{\sqrt 6}, \frac{3}{4 \sqrt 6} \right \}$.

\section{Well and badly approximable numbers} \label{posledni-matematicka}

The search for the best rational approximation of irrational numbers motivates the notion of Markov constant.

\begin{definition}
Let $\alpha$ be an irrational number.
The number
$$
\mu(\alpha) = \inf \left \{ c > 0 \colon \left | \alpha - \frac{k}{m} \right | < \frac{c}{m^2} \text{ has infinitely many solutions } k,m \in \Z  \right \}
$$
is the \textit{Markov constant of $\alpha$}.\\

The number $\alpha$ is said to be {\it well approximable} if $\mu(\alpha) = 0$ and {\it badly  approximable} otherwise.
\end{definition}

We give several comments on the value $\mu(\alpha)$:

\begin{enumerate}
\item Theorem of Hurwitz implies $\mu(\alpha) \leq \frac{1}{\sqrt{5}}$ for any irrational real number $\alpha$.
\item A pair $(k,m)$ which is a solution of $\left | \alpha - \frac{k}{m} \right | < \frac{c}{m^2}$ with  $c\leq \frac{1}{\sqrt{5}}$ satisfies $k=\|m\alpha\|$, where we use the notation $\| x \| = \min \{ | x - n| \colon n \in \Z \}$.
Therefore, \[ \mu(\alpha) = \liminf_{m \to +\infty} m \|m\alpha\| \qquad \text{ and }\qquad
\mu(\alpha) = \min| \S(\alpha) | \]
as the set $\S(\alpha)$ is topologically closed.
\item Due to the inclusion
$\det (g) \S(\alpha) \subset    \S(g\alpha)$ for $g \in \GG$ we can write

$$
|\det (g)|\, \mu(\alpha) \geq \mu(g\alpha).
$$

\item The inequality in \Cref{odhad} implies
\[ \mu(\alpha) = 0 \qquad \Longleftrightarrow \qquad  (a_N)\ \ \ \text{is not bounded }\qquad \Longleftrightarrow \qquad  0\in \S(\alpha). \]
In other words, an irrational number $\alpha$ is well approximable if and only if  the sequence $(a_N)$   of its partial quotients is  bounded.
\end{enumerate}

\subsection{Badly approximable numbers}

As noted above, quadratic irrational numbers serve as an example of badly approximable numbers.
The spectrum $\S(\alpha)$ of such a number has only finite number of elements in the interval $(-\tfrac12, \tfrac12)$.
Theorems \ref{vykousnuto}  and \ref{th:mezi0a1}   give two  examples of spectra of badly approximable numbers of different kinds.

\begin{theorem}\label{vykousnuto}  There exists an irrational number $\alpha$ such that
$\S(\alpha) = (-\infty, -\varepsilon] \cup [\varepsilon, +\infty)$, where $\varepsilon= \frac{\sqrt{2}}{8} \sim 0.18$.
\end{theorem}

We first recall that the natural order on $\R$ is represented by an alternate order in continued fractions.
More precisely, let $x$ and $y$ be two irrational numbers with the continued fractions $[x_0,x_1, \ldots]$ and  $[y_0,y_1, \ldots]$, respectively.
Set $k = \min\{i\in \N \colon x_i\neq y_i\}$.
We have $x<y$ if and only if
$$\bigl(k \ \text{ is even and } \  x_k <y_k\bigr)  \ \ \text{ or } \ \ \  \bigl(k \ \text{ is odd  and } \  x_k >y_k\bigr).$$
To study the numbers with bounded partial quotients we define the following sets:
$$
F(r) = \{ [t, a_1, a_2, \ldots ] \colon t \in \Z, 1 \leq a_i \leq r \}.
$$
and
$$
F_0(r) = \{ [0, a_1, a_2, \ldots ] \colon 1 \leq a_i \leq r \}.
$$
These sets are ``sparse'' and they are Cantor sets: perfect sets that are nowhere dense (see for instance \cite{As99}).
For example, the Hausdorff dimension of $F(2)$ satisfies $0.44 < \dim_H(F(2)) < 0.66$ (see Example 10.2 in \cite{Falconer}).
Taking into account the alternate order, the maximum and minimum elements of $F_0(r)$ can be simply determined.
Thus, $\max F_0(r) = [0,1,r,1,r,1,r,\ldots] $ and $\min  F_0(r) = [0,r,1,r,1,r,1,\ldots] $.
A crucial result which enables us to prove \Cref{vykousnuto} is due to \cite{Ha47} (see also \cite{As99}):
\begin{equation} \label{th:4plus4}
F(4) + F(4) = \R.
\end{equation}
It is worth mentioning that $r=4$ is the least integer for which $F(r) + F(r) = \R$, i.e., in particular, $F(3) + F(3) \neq \R$ (see \cite{Di73}).
Applying  Theorem 2.2 and Lemma 4.2 of \cite{As99} we obtain the following modification of \eqref{th:4plus4}:
\begin{equation}\label{th:4plus40}
F_0(4) + F_0(4) = \left [2  \min F_0({4}), 2\max F_0(4)  \right ] = \left [ \sqrt{2}-1, 4 (\sqrt{2}-1) \right].
\end{equation}
We use the last equality to construct the number $\alpha$ for the proof of \Cref{vykousnuto}.
The construction is based on the following observation.

\begin{lemma} \label{le:na4}
Let ${\bf a} = a_0a_1a_2 \ldots$  be an infinite word over the alphabet  $\mathcal{A} = \{1,2,\ldots,r\}$ such that any finite  string $w_1w_2\cdots w_k$ over the alphabet $\mathcal{A}$ occurs in ${\bf a}$, i.e., there exists index $n \in \N$ such that $a_na_{n+1}\cdots a_{n+k-1} = w_1w_2\cdots w_k$.
Any number $z\in \mathcal{A}+F_0(r) + F_0(r)$ is an accumulation point of the sequence  $(S_{2N})$  and the sequence  $(S_{2N+1})$ with
\begin{equation}\label{ono}S_N=[a_{N+1}, a_{N+2},\ldots] + [0,a_N,a_{N-1},\ldots,a_1].
\end{equation}
\end{lemma}

\begin{proof}  Let $x = [0,x_1,x_2,x_3,\ldots] , y = [0,y_1,y_2,y_3,\ldots]  \in F_0(r)$ and $b\in \A$.
For any  string  $w_1w_2\cdots w_k$ there exist infinitely many finite strings  $u_1u_2 \cdots u_{h-1}u_h$ such that  $w_1w_2\cdots w_k$  is  a prefix and   a suffix of $u_1u_2 \cdots u_{h-1}u_h$.
According to our assumptions each of them occurs at least once in $ {\bf a}$.
It means that any string  $w_1w_2\cdots w_k$ occurs in $ {\bf a}$ infinitely many times on both odd and even positions.
In particular, for any $n$ there exists infinitely many odd and  infinitely many  even indices $N$ such that
$$a_{N-n+1}     \cdots    a_{N}a_{N+1}   \cdots   a_{N+n} = x_nx_{n-1}\cdots x_1by_1y_2\cdots y_{n-2}y_{n-1}.$$
Obviously,  the number $S_N$  given by  \eqref{ono}  equals
\[ b+   [0,y_1,y_2,\ldots, y_{n-1}, a_{N+n}, a_{N+n+1}, \ldots ] + [0,x_1,x_2,\ldots, x_n] . \]
As $b+y+x$ is the limit of the previous sequence, it is an accumulation point of the sequence  \eqref{ono}.
\end{proof}

We can complete the proof of  Theorem \ref{vykousnuto}.
\begin{proof}[Proof of \Cref{vykousnuto}]
We construct an infinite word ${\bf a}$ with letters in $\{1,2,3,4\}$ satisfying the assumptions of \Cref{le:na4}.
We define a sequence $(u_n)_{n=0}^{+\infty}$ recursively as follows: $u_0$ is the empty word and $u_n = u_{n-1}v_n$ where $v_n$ is the word which the concatenation of all words over $\{1,2,3,4\}$ of length $n$ ordered lexicographically.
We have
\begin{align*}
u_1 &= 1 \, 2 \, 3 \, 4 \qquad \text{ and }\\
u_2 &= 1 \, 2 \, 3 \, 4 \, 11 \, 12 \, 13 \, 14 \, 21 \, 22 \, 23 \, 24 \, 31 \, 32 \, 33 \, 34 \, 41 \, 42 \, 43 \, 44 .
\end{align*}
As $u_{n-1}$ is a prefix of $u_n$, we can set ${\bf a}$ to be the unique  infinite word which has a prefix $u_n$ for any $n \in \mathbb{N}$.
One can easily see that ${\bf a}$ satisfies the assumptions of \Cref{le:na4}.

Let  $\alpha$ be the number with the continued fraction $[0,a_1,a_2,a_3,\ldots]$, where $a_1a_2a_3 \ldots = {\bf a}$.
Combining \Cref{odhad,le:na4} and the equality \eqref{th:4plus40} we obtain that
$\tfrac{1}{z}$ and $-\tfrac{1}{z}$  belong to $\S(\alpha)$ for any $z \in [b + {\sqrt{2}-1}, b + 4({\sqrt{2}-1}) ]$ with $b \in \{1,2,3,4\}$.
Overall, we obtain
$$\Bigl[ - \tfrac{1}{\sqrt{2}}, -\tfrac{1}{4\sqrt{2}},\Bigr] \cup \Bigl[ \tfrac{1}{4\sqrt{2}}, \tfrac{1}{\sqrt{2}}\Bigr] \subset \S(\alpha).$$
The property that  $\S(\alpha)$ is closed under multiplication by $z^2$ for each positive integer $z$, in particular under   multiplication by $4$,  already   proves  \Cref{vykousnuto}.
\end{proof}

\begin{theorem} \label{th:mezi0a1}
There exists an irrational number $\alpha$ such that the Hausdorff dimension of $\S(\alpha) \cap \left (-\frac{1}{2}, \frac{1}{2} \right)$ is positive but less than 1.
In particular,  $\S(\alpha) \cap \left (-\frac{1}{2}, \frac{1}{2} \right)$  is  an uncountable set and its   Lebesgue measure is $0$.
\end{theorem}

\begin{proof}
Let $\aa$ be an infinite word with letters in $\{4,5\}$ such that it contains any finite string over $\{4,5\}$ infinitely many times.
A word with such properties can be constructed in the same way as in the proof of \Cref{vykousnuto}.

In accordance with the previous notation we set
\[
F_0(\{4,5\}) = \{ [0, a_1, a_2, \ldots ] \colon a_i \in \{4,5\} \}.
\]
To simplify, we write $F = F_0(\{4,5\})$.
Theorem 1.2 in \cite{As99} implies that
\[
\dim_H ( F+F ) \geq 0.263 \ldots
\]

To obtain an upper bound on the Hausdorff dimension of $F+F$, we first give a construction of $F$.
Let $I$ denote the interval $I = [\min F,\max F]$.
Clearly,   $F\subset I$.

For both letters $z =4$ and $z=5$  we define $f_z: I \to I$ as follows:
\[
f_z(x) = \frac{1}{z + x} \text{ for all } x \in I.
\]
Using the mean value theorem, one can easily derive that
\[
\frac{ | f_z(x) - f_z(y) | }{ | x - y |} \leq \max_{\xi \in I} |f'(\xi)| \leq L := \frac{1}{(\min F+4)^2}
\]
for all $x,y \in I$, $x \neq y$.
Thus, the mappings $f_4$ and $f_ 5$ are contractive and one can see that $F$ is the fixed point of the iterated function system generated by these  mappings.
In other words, we have
\[
F = \lim_{n \to +\infty} Z_n \quad \quad \text{ with } \quad Z_n = \bigcup_{a_1a_2 \cdots a_n \in \{4,5\}^n} f_{a_1} f_{a_2} \cdots f_{a_n} (I).
\]

Let us stress that $\lim_{n \to +\infty}$ on the previous row is defined via the Hausdorff metric on the space of compact subsets of $\R$.

Let $n \in \N$.
It follows that there exists a covering of the set $Z_n$ consisting of $2^n$ intervals of length at most $|I|\cdot L^n$.
Similarly, the set $Z_n + Z_n$ can be covered by $4^n$ intervals of length at most $ |I| \cdot L^n$.
Since $F + F = \lim_{n \to +\infty} Z_n + \lim_{n \to +\infty} Z_n = \lim_{n \to +\infty} ( Z_n + Z_n)$ and $Z_{n+1} \subset Z_n$, we can use this covering to estimate the Hausdorff dimension of $F + F$ (see \cite{Falconer}, Proposition 4.1) as follows:
\[
\dim_H( F + F ) \leq \lim_{n \to +\infty} \frac{ \log 4^n }{ - \log (|I| \cdot L^n) } = -\frac{\log 4}{\log L} = \frac{\log 2}{\log (4+\min F)} .
\]
 As
 $\min F = [0,(5,4)^{\omega}]$, we obtain $\min F = \frac{1}{5+\tfrac{1}{4+\min F}}$.
 Thus $\min F = 2\bigl( \sqrt{\tfrac{6}{5}}-1\bigr)$ and we deduced the upper bound
\[
\dim_H( F + F ) \leq    \frac{\log 2}{\log 2 + \log \bigl(\sqrt{\tfrac{6}{5}}+1\bigr)} <\tfrac12 .
\]

The rest of the proof is analogous to the end of the proof of \Cref{vykousnuto}.
We use \Cref{odhad} and an analogous modification of \Cref{le:na4} for the alphabet $\A = \{4,5\}$ to obtain that
\[
\pm \frac{1}{x}  \in  \S(\alpha)\quad \text{for each } \  x \in \{4,5\} + F + F .
\]
By \Cref{th:12-12} we have
\[
\Bigl\{ \frac{1}{x} \colon  \ | x| \in \{4,5\} + F + F  \Bigr\} \cap \left (-\frac{1}{2}, \frac{1}{2} \right)= \S(\alpha) \cap \left (-\frac{1}{2}, \frac{1}{2} \right).
\]
Clearly, the union of the four sets $4+F + F$,   $5+F + F$, $-4-F - F$, and  $-5-F - F$   with the same Hausdorff dimension is a set of the same dimension.
Moreover,  the Hausdorff dimensions of    $f(M)$ and $M$ coincide for any  continuous mapping $f$, in particular for $f(x)= \tfrac1x$.   It implies that  the estimates on the Hausdorff dimension of $F+F$ are valid also for $\S(\alpha) \cap \left (-\frac{1}{2}, \frac{1}{2} \right)$.
\end{proof}

\subsection{Well approximable numbers}

We consider $\alpha = [a_0,a_1,a_2,\ldots]$  with unbounded partial quotients.
Using second convergents defined in \eqref{secondconvergents} we can write  for any $N\in \N$ and $a\in \N$ with $1 \leq a <a_{N+1}$
\begin{equation}\label{secondUse}
({aq_N+q_{N-1}})^2\Bigl(\frac{ap_N+p_{N-1}}{aq_N+q_{N-1}} - \alpha \Bigr) =(-1)^{N+1} \Bigl(a+ \frac{q_{N-1}}{q_{N}}\Bigr)  \frac{\alpha_{N+1}-a}{\alpha_{N+1} +\frac{q_{N-1}}{q_{N}}} .
\end{equation}
Recall that $\alpha_{N+1} = [a_{N+1}, a_{N+2},a_{N+3}, \ldots]$.
Let $(i_N)$ be a strictly increasing sequence of integers such that $\lim_{N\to +\infty} a_{1+i_N} = +\infty$.
Clearly, $\lim_{N\to +\infty}\alpha_{1+i_N} = +\infty$.
Let us fix $a \in \N$ and put $k_N=ap_{i_N}+p_{i_N-1}$ and $m_N=aq_{i_N}+q_{i_N-1}$, we have
\[
m_N^2 \left| \frac{k_N}{m_N} - \alpha \right| =  \Bigl(a+ \frac{q_{i_N-1}}{q_{i_N}}\Bigr)  E_N
\]
where we set $ E_N = \frac{\alpha_{i_N+1}-a}{\alpha_{i_N+1} +\frac{q_{i_N-1}}{q_{i_N}}}$.
Obviously, $\lim\limits_{N\to +\infty}E_N = 1$.
Since the sequence $(q_N)$ is a strictly increasing sequence of integers, the ratio    $ \frac{q_{i_N-1}}{q_{i_N}}$ belongs to  $(0,1)$.
This implies that the sequence $m_N^2 \left| \frac{k_N}{m_N} - \alpha \right|$ has at least one accumulation point in the interval $[a,a+1]$.
Therefore we can conclude the next lemma.

\begin{lemma}
Let $\alpha$ be an irrational well approximable number.
For any $n \in \N$ the interval $[n,n+1]$ or the interval $[-n-1,-n]$ has a non-empty intersection with $\S(\alpha)$.
\end{lemma}

\newcommand{\euler}{ \mathrm e}
\begin{example}
Unlike  the number $\pi$, the continued fraction of the Euler constant  has a regular structure:
$$\euler = [2, 1, 2, 1, 1, 4, 1, 1, 6, 1, 1, 8, 1, 1, 10, \ldots].$$
In general, for $\euler = [2,a_1,a_2, \ldots ]$ we have
\[
a_{3n+1}=1,  \quad   a_{3n+2}=2(n+1) \quad \text{ and } \quad a_{3n+3}=1  \quad \text{ for any } n \in \N.
\]
We demonstrate that
$$
\left|q_{3N}^2 \Bigl( \frac{p_{3N}}{q_{3N}} - \mathrm e \Bigr)\right| \to \tfrac12 .
$$
By Lemma \ref{odhad} we need to show
$$ A_{3N} :=[a_{3N+1}, a_{3N+2},a_{3N+3},\ldots] + [0,a_{3N},a_{3N-1},\ldots,a_1] \to 2 .$$
Using the simple estimate valid for any  continued fraction
$$
 b_0 + \frac{1}{b_1+ \frac{1}{b_2}}\ < \  [b_0,b_1,b_2,b_3,\ldots]\  < \ b_0 + \frac{1}{b_1+ \frac{1}{1+b_2}}
$$
we obtain the following bounds:
$$1+ \frac{1}{2(N+1) + 1} + \frac{1}{1+ \frac{1}{{2N}} }< A_{3N} <   1+ \frac{1}{{2(N+1)}+ \frac{1}{2}} + \frac{1}{1+ \frac{1}{{2N}+1} } .$$
Both bounds have the same limit, namely $2$,  as we wanted to show.
Analogously one can deduce that
$$
\left|q_{3N-1}^2 \Bigl( \frac{p_{3N-1}}{q_{3N-1}} - \mathrm e \Bigr)\right|\to \tfrac12   \quad \text{ and } \quad  \left|q_{3N+1}^2 \Bigl( \frac{p_{3N+1}}{q_{3N+1}} - \mathrm \euler \Bigr)\right|\to 0 .
$$
Since $(-1)^{3N}$ takes positive and negative signs, the values $0, \pm \tfrac12$ belong to the spectrum of $\euler$ and moreover
$$(-\tfrac12, \tfrac12)\cap \S(\euler)=\{0\} .$$
As $a_{3N+2} =2N> 1$, we can use the second convergents as well and for any fixed  $a \in \mathbb{N}$ and  any $N$ such that $a < a_{3N+2}$  we obtain
$$({aq_{3N+1}+q_{3N}})^2\Bigl|\euler - \frac{ap_{3N+1}+p_{3N}}{aq_{3N+1}+q_{3N}}\Bigr| = \Bigl(a+ \frac{q_{3N}}{q_{3N+1}}\Bigr)  E_{3N},
$$
where   $\lim\limits_{N\to \infty}E_{3N} = 1$, cf. \eqref{secondUse}.
By the proof of \Cref{odhad}, we have
$$ \frac{q_{3N}}{q_{3N+1}} =  [0,a_{3N+1},a_{3N}, a_{3N-1},\ldots,a_1]  \to \tfrac12. $$
We conclude  for the spectrum of the Euler number satisfies
$$ \{0\} \cup \{a+\tfrac12: a \in \mathbb{Z}\}\subset \S(\euler)\,.$$
Of course, the inclusion cannot be replaced by an equality.
The reason is simple; the spectrum is closed under multiplication by the factor $4$,  and thus
$$\{4a+2: a \in \mathbb{Z}\}\subset \S(\euler)$$
as well.
\end{example}

\begin{theorem}\label{vse}  There exists an irrational number $\alpha$ such that
$\S(\alpha) =\R$.
\end{theorem}

\begin{proof}  Suppose that ${\bf a}=a_1a_2\ldots $ is an infinite word such that any   sequence of  the form $w_1w_2\ldots w_kNw_{k+1}w_{k+2} \cdots w_{2k}$,  where the symbols  $w_i$  are from  the alphabet  $\{1,2,3,4\}$
 and $N>1, N\in \Z$   occurs in ${\bf a}$ infinitely many times.
 The same reasoning as in the  proofs of \Cref{le:na4} and \Cref{vykousnuto}   together with the equality \eqref{th:4plus40} imply the statement of the theorem.
 Therefore, it is enough to describe ${\bf a}$.

Fix $n \in \mathbb{N}$ and consider a word $w=w_1w_2\cdots w_n$ of length $n$ over the alphabet $\{1,2,3,4\}$.
${\rm Copy}(w)$ denotes the concatenation of $n$ words of length $(n+1)$, each in the form  $wh=w_1w_2\cdots w_nh$ with $h=1,2,\ldots, n$. Thus ${\rm Copy}(w)$ is a word of length $n(n+1)$.
The word $v_n$ is created by concatenation of ${\rm Copy}(w)$ for all  words $w$ of length $n$ over the alphabet      $\{1,2,3,4\}$.
In particular, the length of $v_n$ is $4^nn(n+1)$.

The infinite word ${\bf a}$ is given by its prefixes $(u_n)$ which are constructed recursively:   $u_0$ is   the empty word and $u_n = u_{n-1}v_n$.
\end{proof}

\begin{remark}
We note that the behaviour of $\alpha = [a_1, a_2, \ldots]$ defined in the proof of the previous theorem is typical.
In  \cite{BoJaWi83},  Bosma, Jager and Wiedijk described the distribution of the sequence $q_n|p_n - \alpha q_n|$.
A direct consequence of their result is  that $\S(\alpha) =\R$ for almost all $\alpha \in [0,1]$.
Thus, it is also true for almost all $\alpha \in \R$.
\end{remark}

\section{Discussion and remarks} \label{konec}

In the context of physics it must be emphasized that our choice of the elementary model (\ref{presm}) + (\ref{represm}) is motivated not only by its appealing number-theoretical properties but also by its possible straightforward phenomenological applicability.
We feel motivated by the persuasion that the related {\em constructive} exemplification of certain spectral anomalies might prove attractive even from the point of view of a physicist who need not necessarily care about the deeper mathematical subtleties.

Using our purely mathematical tools we are able to arrive at a better understanding of certain purely formal connections between various structural aspects of the spectra, with the main emphasis put on its unboundedness from below (which could result into instabilities under small perturbations) in an interplay with the emergence of accumulation points in the {point spectrum} (in the latter case it makes sense to keep in mind the existing terminological ambiguities \cite{SK}).

Needless to add that the phenomenological role of the spectral accumulation points remains strongly model-dependent (see the rest of this section for a few samples).
In the most elementary quantized hydrogen atom, for example, such a point represents just an entirely innocent lower bound of the continuous spectrum.
A more interesting interpretation of these points is obtained in the case of the so called Efimov three-body bound states \cite{Efimov,Efimov2}, etc.

\subsection{The context of systems with position-dependent mass}

Irrespectively of the concrete physical background of quantum stability \cite{cinani}, its study encounters several subtle mathematical challenges \cite{Behrndt,Behrndt2,Behrndt3,Behrndt4}.
In our present hyperbolic-operator square-well model living on a compact domain $R$, a number of interesting spectral properties is deduced and proved by the means and techniques of mathematical number theory, without any recourse to the abstract spectral theory.
Still, the standard spectral theory is to be recalled.
For example, once we return to the explicit units we may reinterpret our present {\em hyperbolic} partial differential operator $\Box$ in Eq.~(\ref{presm}) as a result of a drastic deformation of an {\em elliptic} non-equal-mass Laplacean
 \begin{equation}
 \triangle= \frac{1}{2m_x}\,\frac{\partial^2}{\partial x^2}
 +\frac{1}{2m_y}\,\frac{\partial^2}{\partial y^2}\,
 \end{equation}
or rather of an even more general kinetic-energy operator
 \begin{equation}
 T(x,y)= \frac{1}{2m_x(x,y)}\,\frac{\partial^2}{\partial x^2}
 +\frac{1}{2m_y(x,y)}\,\frac{\partial^2}{\partial y^2}\,
 \end{equation}
containing the position-dependent positive masses.
In the ultimate and decisive step one simplifies the coordinate dependence in the masses $m_{x,y}(x,y)$ (say, to piecewise constant functions) and, purely formally, allows one of them to become negative.

In such a setting our present mathematical project is also guided by the specific position-dependent mass physical projects of Refs.~\cite{GeZa,levai2009scattering} inspired, in their turn, by the non-Hermitian (a.k.a. ${\cal PT}-$symmetric \cite{0034-4885-70-6-R03}) version of quantum Kepler problem.
In these papers the mass $m(x)$ is allowed to be complex and, in particular, negative.
In \cite{GeZa} the onset of the spectral instability is analyzed as an onset of an undesirable unboundedness of the discrete spectrum from below.
A return to a stable system with vacuum is then shown controllable only via an energy-dependent mass $m(x,E)$, i.e., via an {\em ad hoc} spectral cut-off (cf. also \cite{taky}).

\subsection{The context of generalized quantum waveguides}

Before one recalls the boundary conditions (\ref{represm}), the majority of physicists would perceive our hyperbolic partial differential Eq.~(\ref{presm}) as the Klein-Gordon equation describing the free relativistic one-dimensional motion of a massive and spinless point particle.
Whenever one adds an external (say, attractive Coulomb) field, the model becomes realistic (describing, say, a pionic atom).
Now, even if we add the above-mentioned Dirichlet boundary conditions $f|_{\partial R} = 0$, a certain physical interpretation of the spectrum survives the characterization of, say, the bound states in a ``relativistic quantum waveguide''.

One of the most interesting consequences of the latter approach may be seen in the possibility of a collapse of the system in a strong
field.
The most elementary illustrations of such a type of instability may even remain non-relativistic: Landau and Lifshitz \cite{Landau} described the phenomenon in detail.
Another, alternative, type of quantum instabilities connected with the emergence of spectral accumulation points occur also in Horava-Lifshitz gravity with ghosts \cite{Hawking,ghost} or in the conformal theories of gravity \cite{ghostc,conformalgra,ghostd} etc.

Our present choice of the elementary illustrative example with compact and rectangular $R$ changes the physics and becomes more intimately related to the problems of the so called quantum waveguides with impenetrable walls \cite{Exnerkniha}.
Most of the mathematical problems solved in the latter context are very close to the present ones.
Typically, they concern the possible relationship between the spectra and geometry of the spatial boundaries.
In this setting, various transitions to the infinitely thin and/or topologically nontrivial domains $R$ (one may then speak about quantum graphs) and, possibly, also to the various anomalous point-interaction forms of the interactions are being also studied \cite{hexe}.

Up to now, people only very rarely considered a replacement of the positive-definite kinetic-energy operator (i.e., Laplacean) by its hyperbolic alternative.
Thus, in spite of some progress \cite{Davidi}, such a ``relativistic'' generalization of the concept of quantum waveguide and/or of quantum graph still remains to be developed.

\subsection{The context of classical optical systems with gain and loss} \label{opreceding}

One of the most characteristic features of modern physics may be seen in the multiplicity of overlaps between its apparently remote areas. {\it Pars pro toto} let us mention here the unexpected productivity of the transfer of several quantum-theoretical concepts beyond the domain of quantum theory itself \cite{MZbook}.
One of the best known recent samples of such a transfer starts in quantum field theory \cite{PhysRevD.55.R3255} and ends up in classical electrodynamics \cite{0295-5075-81-1-10007}.
A common mathematical background consists in the requirements of the Krein-space self-adjointness \cite{Kuzhelinbook} {\it alias} parity-times-time-reversal symmetry (${\cal PT}-$symmetry).

It is worth adding that the latter form of a transfer of ideas already proceeded in both directions.
The textbook formalism of classical electrodynamics based on Maxwell equations was enriched by the mathematical techniques originating in spectral theory of quantum operators in Hilbert space (cf., e.g., section 9.3 of the review paper \cite{doi:10.1142/S0219887810004816} for more details).
In parallel, the ${\cal PT}-$symmetry-related version of quantum theory (cf. also its older review \cite{0034-4885-70-6-R03}) took an enormous profit from the emergence and success of its experimental tests using optical metamaterials \cite{meat}. People discovered that the time is ripe to think about non-elliptic versions of Maxwell equations reflecting the quick progress in the manufacture of various sophisticated metamaterials which possess non-real elements of the permittivity and/or permeability tensors \cite{DeLa04,Kriegler,Tanaka,Soukolis}.

Naturally, the mutual enrichments of the respective theories would not have been so successful without the progress in experimental techniques, and {\it vice versa}.
In fact, the availability of the necessary optical metamaterials (which can simulate the ${\cal PT}-$symmetry of quantum interactions via classical gain-loss symmetry of prefabricated complex refraction indices) was a highly nontrivial consequence of the quick growth of the know-how in nanotechnologies \cite{BAS88,ShSmSc}.
In opposite direction, the experimental simulations of various quantum loss-of-stability phenomena in optical metamaterials encouraged an intensification of the related growth of interest in the questions of stability of quantum systems with respect to perturbations \cite{Davies_2007,Chen20081986,Zn12}.

\subsection{The context of unbounded spectra}

Our last comment on the possible phenomenological fructification of our present study of the toy model (\ref{presm}) + (\ref{represm}) concerns its possible, albeit purely formal, connection to the traditional Pais-Uhlenbeck (PU) oscillator \cite{PU}.
The idea itself is inspired by the Smilga's paper \cite{Smilgasigma} which provides us with a compact review of the appeal of the next-to-elementary PU model in physics.

We imagine, first of all, that the unboundedness of the spectrum of the PU oscillators parallels the same ``threat of instability'' feature of our rectangular model.
At the same time, in the broad physics community, the PU oscillator is much more widely accepted as a standard model throwing a new light on several methodical aspects of the loss of stability, especially in the context of quantum cosmology and quantization of gravity (cf. also \cite{pugra,pugrab,pugrac}).
In particular, the PU model contributes to the understanding of the role of renormalizability in higher-order field theories~\cite{Hawking,Kaparulin}, etc.

For these reasons we skip the problems connected with the ambiguity of transition from Lagrangians to Hamiltonians \cite{Ruzi} and we restrict our attention just to one of the specific, PU-related quantum Hamiltonian(s), viz., to the operator picked up for analysis, e.g., in Ref.~\cite{Smilgasigma},
 \begin{equation}
 \label{Hdiag}
 H  =  \left ({-\partial_x^2 + \Omega_x^2 x^2}\right )
  -  \left ({-\partial_y^2 +
 \Omega_y^2 y^2} \right ).
  \end{equation}
In a way resembling our present results, the related quantum dynamics looks pathological because even the choice of the incommensurable oscillator frequencies $\Omega_x$ and $\Omega_y$ leads to a quantum system in which the bound-state energy spectrum (i.e., in the language of mathematics, point spectrum - cf. a comment Nr. 2 in \cite{Smilgasigma}) is real but dense and unbounded,
 \begin{equation}
  \label{spec12}
  E_{nm} =
 \left(n+ \frac 12 \right) \Omega_x - \left( m + \frac 12 \right)
 \Omega_y, \qquad n,m = 0,1,2,\ldots .
 \end{equation}
In the related literature (cf., e.g., \cite{fabio,fabiob,fabioc}) several remedies of the  pathologies are proposed ranging from the use of the Wick rotation of $y \to {\rm i}y$ \cite{BeMa} up to a suitable modification of the Hamiltonian as performed already before quantization, on classical level \cite{Leach,Leachb,Leachc}.

This being said, an independent disturbing feature of the PU toy model (\ref{Hdiag}) may be seen in an abrupt occurrence of a set of spectral accumulation points in the equal-frequency limit $\Omega_x-\Omega_y \to 0$ \cite{alipu,alipub}.
The emergent new technical difficulty originates from the fact that the resulting Hamiltonian becomes non-diagonalizable, acquiring a rather peculiar canonical matrix structure of an infinite-dimensional Jordan block.

This is one the most dangerous loss-of-quantum-meaning aspects of the model.
Its serious phenomenological consequences are discussed, e.g., in the scalar field cosmology (cf. the freshmost papers \cite{portugalci,portugalcib} with further references).
In a narrower context of specific pure fourth-order conformal gravity, such a spectral discontinuity cannot be circumvented at all \cite{Maconfor}.

 \section{Summary} \label{ukonec}

The aim of this  paper is to demonstrate the variability of spectra in dependence on the number-theoretical properties of the ratio $\alpha=a/b$ of the sides of the rectangular $R$.
In particular, we show that in an arbitrarily short interval $I \subset \R$ one can find   numbers $\alpha, \beta, \gamma, \delta, \varepsilon$ such that the spectrum of $\S(\alpha)$ is empty, the spectrum of $\S(\beta)$ forms an infinite discrete set,  the spectrum $\S(\gamma) = \R$ covers the whole real line, the spectrum $\S(\delta) = \R\setminus (-a,a)$ has a ``hole'' with some positive real $a=a(\delta)$.
Finally  the spectrum  $\S(\varepsilon)$ has zero Lebesgue measure, it is uncountable, but it has a positive Hausdorff dimension which is less than 1.
It means that a small change of the dynamical parameter $\alpha=a/b$ dramatically influences  the spectrum.

Although we give just an extremely elementary example for the detailed and rigorous analysis, we would like to emphasize that our present approach proves productive in spite of lying far beyond the standard scope and methods of spectral analysis.
A nontrivial insight in the underlying physics is provided purely by the means of number theory.

From the point of view of number theory, various results on the Markov constant, i.e., $\min \left \{ m^2 \left| \frac{k}{m} - \alpha \right| \colon k,m \in \Z  \right \}$, may be found.
In the present article, we provide some insight into the behaviour of all the accumulation points of the concerned set.
Since we restrict ourselves to some special cases, naturally, the next step would be to fully investigate the properties of $\S(\alpha)$.

\section*{Acknowledgements}

EP acknowledges financial support from the Czech Science Foundation
grant 13-03538S and \v{S}S acknowledges financial support from the Czech
Science Foundation grant 13-35273P. MZ was supported by RVO61389005
and by GA\v{C}R Grant 16-22945S.


\providecommand{\newblock}{}

\end{document}